\PassOptionsToPackage{hyperfootnotes=false}{hyperref}
\documentclass[version=preprint]{iacrcc}

\usepackage[utf8]{inputenc}
\usepackage[english]{babel}
\usepackage{booktabs}
\usepackage{enumitem}
\usepackage{algorithm}
\usepackage[noend]{algpseudocode}

\makeatletter
\providecommand*{\theHALG@line}{}
\renewcommand*{\theHALG@line}{\thealgorithm.\arabic{ALG@line}}
\makeatother

\allowdisplaybreaks[2]

\numberwithin{theorem}{section}
\makeatletter
\let\c@lemma\c@theorem
\let\c@proposition\c@theorem
\let\c@corollary\c@theorem
\let\c@definition\c@theorem
\let\c@problem\c@theorem
\let\c@remark\c@theorem

\let\p@lemma\p@theorem
\let\p@proposition\p@theorem
\let\p@corollary\p@theorem
\let\p@definition\p@theorem
\let\p@problem\p@theorem
\let\p@remark\p@theorem

\makeatother

\newcommand{\F}{\mathbb F}
\newcommand{\Q}{\mathbb Q}
\newcommand{\Z}{\mathbb Z}
\newcommand{\OO}{\mathcal O}
\newcommand{\DD}{\mathcal D}
\newcommand{\PP}{\mathfrak P}
\newcommand{\idealA}{\mathfrak a}
\newcommand{\idealB}{\mathfrak b}
\DeclareMathOperator{\End}{End}

\DeclareMathOperator{\Pic}{Pic}
\DeclareMathOperator{\Trd}{Trd}
\DeclareMathOperator{\Nrd}{Nrd}
\DeclareMathOperator{\cond}{cond}
\DeclareMathOperator{\disc}{disc}
\newcommand{\ceil}[1]{\left\lceil #1\right\rceil}
\newcommand{\floor}[1]{\left\lfloor #1\right\rfloor}
\newcommand{\OneEnd}{\textsc{OneEnd}}
\newcommand{\EndRing}{\textsc{EndRing}}
\newcommand{\Isogeny}{\textsc{Isogeny}}
\newcommand{\Supp}{\operatorname{Supp}}

\title[
  running={The supersingular isogeny problem},
  plaintext={The Supersingular Isogeny Problem in Time and Memory p-to-the-(1/3+o(1)), Unconditionally}
]{The Supersingular Isogeny Problem in Time and Memory
  \(p^{1/3+o(1)}\), Unconditionally}

\addauthor[
  inst={1},
  email={jose@delgado.fyi},
  surname={Delgado}
]{José Luis Delgado}

\addaffiliation{Independent Researcher}

\begin{document}
\maketitle

\begin{abstract}
Given a supersingular elliptic curve \(E/\F_{p^2}\), the \OneEnd\ problem asks
for a non-scalar endomorphism of \(E\).  By known
reductions, solving this problem also solves the supersingular endomorphism
ring and isogeny problems.  Weso{\l}owski obtained exponent \(1/3\) under an
assumption on the factorization of a small degree, whereas the previous
unconditional exponent was \(2/5\).  We give a Las Vegas algorithm, analyzed
without a smoothness heuristic, with expected time and memory
\[
p^{1/3}\exp\bigl(O(\sqrt{\log p\,\log\log p})\bigr)
 =p^{1/3+o(1)}.
\]

The algorithm fixes in advance a family of degrees that are products of small
primes.  Known counting results provide many isogenies of these degrees from
curves to their Frobenius conjugates, and a collision estimate shows that the
isogenies occur on sufficiently many distinct curves for a random walk to
reach one of them.  From such a curve, the algorithm splits a degree into two
parts, enumerates two lists of shorter isogenies, and matches their targets to
obtain an isogeny to the conjugate, whose composition with Frobenius gives the
required endomorphism.
\end{abstract}

\begin{textabstract}
Given a supersingular elliptic curve E over F_(p^2), the OneEnd problem asks
for a non-scalar endomorphism of E. By known reductions, solving this problem
also solves the supersingular endomorphism ring and isogeny problems.
Wesolowski obtained exponent 1/3 under an assumption on the factorization of a
small degree, whereas the previous unconditional exponent was 2/5. We give a
Las Vegas algorithm, analyzed without a smoothness heuristic, with expected
time and memory p^(1/3) exp(O(sqrt(log p log log p))) = p^(1/3+o(1)). The
algorithm fixes in advance a family of degrees that are products of small
primes. Known counting results provide many isogenies of these degrees from
curves to their Frobenius conjugates, and a collision estimate shows that the
isogenies occur on sufficiently many distinct curves for a random walk to
reach one of them. From such a curve, the algorithm splits a degree into two
parts, enumerates two lists of shorter isogenies, and matches their targets to
obtain an isogeny to the conjugate, whose composition with Frobenius gives the
required endomorphism.
\end{textabstract}

\section{Introduction}
\label{sec:introduction}

\subsection{The problem}

For a supersingular elliptic curve \(E/\F_{p^2}\), the \OneEnd\ problem asks
for a non-scalar endomorphism of \(E\).
Unconditional reductions connect
it to the computation of the full ring \(\End(E)\) and of an isogeny between
two supersingular curves~\cite{PW24,MW25}, so an improved algorithm for
\OneEnd\ yields improved algorithms for all three problems.

These problems govern the security of several cryptographic constructions
based on supersingular isogenies, and their generic classical cost was
\(p^{1/2+o(1)}\) for many years.  Weso{\l}owski reduced the exponent to
\(1/3\) under a heuristic about smooth integers~\cite{Wes26}, and Udovenko
later obtained an unconditional exponent of \(2/5\)~\cite{Udovenko26}.  The
question left between these results is:
\[
 \textit{Can the exponent \(1/3\) be reached without a smoothness
 assumption?}
\]
The algorithm developed here attains this exponent unconditionally, using the
same starting observation as both preceding algorithms.  Write \(E^{(p)}\)
for the curve obtained by applying the
\(p\)-power map to the coefficients of \(E\).  Given a separable isogeny
\[
 \varphi:E\longrightarrow E^{(p)}.
\]
The relative Frobenius \(\pi_{E^{(p)}}:E^{(p)}\to E\) then gives
\[
 \pi_{E^{(p)}}\circ\varphi\in\End(E).
\]
The resulting endomorphism has inseparable degree \(p\), whereas an integer
multiplication has inseparable degree with even \(p\)-adic valuation, so the
displayed endomorphism is not scalar.

The remaining task of finding \(\varphi\) begins with the theorem of Aubry,
Oyono, and Vincent, which guarantees such an isogeny of degree at most
\((p/2)^{1/3}\)~\cite{AOV26}, but their existence theorem specifies the size
of the degree and leaves its factorization unspecified.  A large prime factor
makes a direct search too expensive, and Weso{\l}owski handles this obstruction by
assuming that the least possible degree behaves like a random integer and is
smooth often enough.  The argument below attains the same exponent through a
family of degrees with prescribed factorizations.

\subsection{The basic search}

Consider first a degree \(d=\ell_1\ell_2\), where \(\ell_1\) and \(\ell_2\)
are distinct small primes, and suppose that an isogeny
\(\varphi:E\to E^{(p)}\) of degree \(d\) factors through an intermediate
curve \(C\):
\[
 E\xrightarrow{\lambda}C\xrightarrow{\eta}E^{(p)},
 \qquad \deg\lambda=\ell_1,\quad\deg\eta=\ell_2.
\]
With \(\rho=\widehat\eta:E^{(p)}\to C\), a list of the degree \(\ell_1\)
isogenies from \(E\) and a list of the degree \(\ell_2\) isogenies from
\(E^{(p)}\) both contain \(C\), and matching those entries recovers
\[
 \varphi=\widehat\rho\circ\lambda.
\]

The same construction works when a suitable degree \(d\) factors as
\[
 d=uv,\qquad \gcd(u,v)=1,
\]
with \(u\) and \(v\) of comparable size.  For a cyclic isogeny
\(\varphi:E\to E^{(p)}\) of degree \(d\), splitting its prime factors between
\(u\) and \(v\) produces a curve \(C\) and isogenies
\[
 \lambda:E\longrightarrow C,\qquad
 \rho:E^{(p)}\longrightarrow C
\]
of degrees \(u\) and \(v\).  They again recover the original map through
\[
 \varphi=\widehat\rho\circ\lambda.
\]

This factorization gives a direct algorithm: list the isogenies of degree at
most about \(\sqrt d\) from both \(E\) and \(E^{(p)}\), then match entries
whose targets are isomorphic.  Each match recovers \(\varphi\), and the work
is essentially the combined size of the two lists.

The factorization of \(d\) determines the cost of this search: when \(d\) is
a product of small primes, both lists can be constructed by following
isogenies of those prime degrees; an unknown prime factor of size close to
\(d\) requires a search beyond the target cost.

Weso{\l}owski applies this search to the least degree from \(E\) to
\(E^{(p)}\), so its success depends on the factorization of one integer.  Our
construction first chooses a large family \(\mathcal D_p\) in which every
degree is squarefree and smooth, and then proves that many supersingular
curves admit an isogeny to their conjugate whose degree belongs to this
family.

\subsection{Coverage by the chosen degrees}

For one degree \(d\), results of Chenu and Smith count many pairs consisting
of a curve and a cyclic degree \(d\) isogeny to its conjugate~\cite{CS22}.
The number of such pairs is of order \(\sqrt{pd}\), up to logarithmic factors.
Summing this count over \(d\in\mathcal D_p\) gives enough pairs for the
algorithm, provided that they are distributed among sufficiently many
curves.  Concentration on a small set would reduce the probability that a
random walk finds a useful curve, so the analysis must bound how often two
degrees \(d\) and \(e\) occur on the same curve.

For two isogenies
\[
 \varphi,\psi:E\longrightarrow E^{(p)}
\]
of squarefree degrees \(d\) and \(e\), consider the composite
\[
 \gamma=\widehat\varphi\circ\psi\in\End(E)
\]
which has norm \(de\).  Distinct kernels give a non-scalar \(\gamma\), whose
trace \(t\) satisfies
\[
 t^2<4de.
\]
The element \(\gamma\) generates an imaginary quadratic field and determines
the largest quadratic order in that field that embeds in \(\End(E)\) through
\(\gamma\).

The first isogeny becomes an invertible ideal of norm \(d\) in this order,
and the existence of the second isogeny imposes an additional divisibility
condition on that ideal.  Once the ideal is fixed, the common endpoint of the
two isogenies satisfies a square equation in the class group of the order;
every nonempty fiber of this square map has the size of the \(2\)-torsion
subgroup, which genus theory bounds by a divisor function.

Every pair of isogenies therefore yields arithmetic data of the stated form,
and counting all compatible data gives an upper bound without requiring a
converse construction.  Summation over \(t\) and over the possible orders
produces a bound of size \(\sqrt{de}\), apart from logarithmic and local
factors, which prevents the isogenies counted by Chenu and Smith from
concentrating on too few curves.

\subsection{Main result}

\begin{theorem}[Main theorem]\label{thm:intro-main}
Let \(p>3\) be prime.  There is a classical Las Vegas algorithm which, given a
supersingular elliptic curve \(E/\F_{p^2}\), returns a non-scalar endomorphism
\(\alpha\in\End(E)\setminus\Z\), in efficient representation, in expected time
and memory
\[
p^{1/3}\exp\bigl(O(\sqrt{\log p\,\log\log p})\bigr)
 =p^{1/3+o(1)}.
\]
The returned endomorphism satisfies \(\log\deg\alpha=O(\log p)\), and the
algorithm and its analysis are unconditional.
\end{theorem}

\begin{corollary}\label{cor:intro-equivalent}
The supersingular endomorphism ring problem and the supersingular isogeny
problem admit classical probabilistic algorithms with expected time and
memory \(p^{1/3+o(1)}\).
\end{corollary}

The proof rests on an estimate for the distribution of the chosen degrees
among curves.  In the terminology of Chenu and Smith, a
\((d,+1)\)-structure is a separable cyclic isogeny
\(\varphi:E\to E^{(p)}\) of degree \(d\) such that
\(\widehat\varphi=\varphi^{(p)}\).  For a generic model \(E\), let \(r_d(E)\)
count the cyclic kernels that support such an isogeny.  If \(d,e\geq2\) are
squarefree and prime to \(p\), then
\[
 \sum_E r_d(E)r_e(E)
 -\mathbf1_{d=e}\sum_Er_d(E)
 \ll
 \sqrt{de}\,\log^2(2de)
 \prod_{\ell\mid\gcd(d,e)}
 \left(1+\frac3\ell\right).
\]
The implied constant is absolute, the subtraction removes the repeated use
of one kernel, and the product records the extra choices at primes shared by
the two degrees.  Apart from logarithms and the displayed product, the bound
grows like \(\sqrt{de}\), which prevents the isogenies from concentrating on
a few curves.  Its proof assigns arithmetic data to each pair of isogenies
and counts the resulting data; Theorem~\ref{thm:collision} gives the finite
form of the estimate.

\subsection{Algorithmic use of the collision bound}

For \(n=\lfloor\log_2p\rfloor+1\), Section~\ref{sec:density} specifies the
primes and the number of factors in a set \(\mathcal D_p\) of squarefree
degrees chosen before sampling.  Every prime factor is small enough for direct
enumeration, and every member satisfies
\[
 d\leq
p^{1/3}\exp\bigl(O(\sqrt{\log p\,\log\log p})\bigr).
\]

The Chenu--Smith count gives the total number of isogenies for each \(d\),
Tatuzawa's theorem supplies the required lower bound for every degree except
possibly one, and the collision theorem controls how often two counted
isogenies occur on the same curve.  Together these estimates show that an
inverse polynomial proportion of the generic supersingular models supports
the required isogenies, with the numerical bound of
Theorem~\ref{thm:density}.

Using these estimates, each trial of the algorithm performs the following five
steps.

\begin{enumerate}[leftmargin=2em,label=\textup{(\arabic*)}]
\item Move from the input along a lazy walk in the \(2\)-isogeny graph until,
      after \(O(\log p)\) steps, the endpoint is close to the stationary
      distribution.
\item At the endpoint \(E\), list the relevant isogenies from \(E\) and from
      \(E^{(p)}\), and split every candidate degree into two parts of
      comparable size.
\item Match records that reach isomorphic curves, using a coloring of the
      prime factors to enforce disjoint supports without testing every pair
      of records.
\item Compose the matched records to obtain
      \(\varphi:E\to E^{(p)}\), and verify every map and isomorphism used in
      the construction.
\item Form
      \(\widehat\omega\circ\pi_{E^{(p)}}\circ\varphi\circ\omega\), where
      \(\omega\) is the initial walk, and convert the result to an efficient
      representation.
\end{enumerate}

A trial without a match is repeated, whereas every returned map has passed
deterministic checks; randomness therefore affects the running time but not
correctness.

\subsection{Relation to previous algorithms}

The outer construction follows Algorithms~2 and~3 of~\cite{Wes26}, which walk
to a new curve, search from the curve and its conjugate, compose the result
with Frobenius, and return along the walk.  Our algorithm retains these steps
and replaces the heuristic source of a smooth degree with the family
\(\mathcal D_p\) and the collision bound.

Algorithm~1 of~\cite{Udovenko26} uses the same walk, Frobenius, and return
path, and detects one chosen degree through a modular polynomial and bivariate
multipoint evaluation.  Our argument uses the Chenu--Smith count in the form
developed there.

The theorem is asymptotic, with a large subexponential factor and memory of
the same order as the running time; practical attack costs and concrete
security levels lie outside its scope.  The reduction to \Isogeny\ permits
an unrestricted output degree, while a prescribed prime degree lies outside
the statement.

\subsection{Organization}

\Cref{sec:prelim} fixes the computational model and notation;
\cref{sec:collisions} proves the collision bound; \cref{sec:density}
constructs \(\mathcal D_p\) and proves that many curves support its degrees;
and \cref{sec:algorithm} derives the algorithm and its running time from that
density statement.

\section{Preliminaries}
\label{sec:prelim}

\subsection{Elliptic curves and isogenies}

The counting argument takes place over \(\F_{p^2}\), with any other field
displayed explicitly, and all endomorphism rings are geometric.  The choice
of input model in \cref{sec:input-normalization} is the sole use of an
extension, whose degree is at most two.

An isogeny \(\varphi:E\to E'\) is a nonconstant rational homomorphism of
groups, and its degree is the degree of the rational map.  It is
\emph{separable} when \(\deg\varphi=\#\ker\varphi\), and \emph{cyclic} when
its geometric kernel is cyclic.  Each isogeny has a unique dual
\(\widehat\varphi:E'\to E\) satisfying
\[
 \widehat\varphi\circ\varphi=[\deg\varphi],
 \qquad
 \varphi\circ\widehat\varphi=[\deg\varphi].
\]
Maps are composed from right to left, so \(\psi\circ\varphi\) applies
\(\varphi\) first.

An elliptic curve in characteristic \(p\) is supersingular if it has no
nonzero geometric \(p\)-torsion; it then has a model over \(\F_{p^2}\), and
its endomorphism ring is a maximal order in the quaternion algebra over
\(\Q\) ramified at \(p\) and infinity.  We write
\[
 \End^0(E)=\End(E)\otimes_{\Z}\Q.
\]
The map \(a\mapsto[a]\) identifies \(\Z\) with the center of \(\End(E)\),
whose elements are the \emph{scalar} endomorphisms; an element of
\(\End(E)\setminus\Z\) is \emph{non-scalar}.  On \(\End(E)\), the degree is
the reduced norm, and the reduced trace of \(u\) is \(u+\widehat u\).

\subsection{Efficient representations and computational problems}

Because the degree of an isogeny may be exponential in the input length,
writing the whole rational map can require too much space, so we use the output
convention of~\cite{Wes26,MW25}.

\begin{definition}[Efficient representation]\label{def:efficient}
An algorithm \(\mathsf A\) that runs in polynomial time is an \emph{efficient
isogeny evaluator} if, for every \(D\in\{0,1\}^*\) such that
\(\mathsf A(\mathsf{validity},D)=\top\), there is an isogeny
\(\varphi:E\to E'\), defined over a finite field \(\F_q\), such that
\begin{enumerate}[label=(\arabic*),leftmargin=2.4em,itemsep=0pt,topsep=2pt]
\item \(\mathsf A(\mathsf{curves},D)=(E,E')\);
\item \(\mathsf A(\mathsf{degree},D)=\deg\varphi\);
\item for \(P\in E(\F_{q^k})\),
      \(\mathsf A(\mathsf{eval},D,P)=\varphi(P)\).
\end{enumerate}
If, moreover, the length of \(D\) is polynomial in
\(\log\deg\varphi+\log q\), then \(D\) is an \emph{efficient
representation} of \(\varphi\) with respect to \(\mathsf A\).
\end{definition}

During the search, a map is stored as a chain of isogenies of small degree and
converted by isogeny interpolation into an efficient representation before
it is returned, as proved in \cref{sec:output-representation}.

As in~\cite{Wes26}, all isogenies returned in the following computational
problems are encoded in an efficient representation.

\begin{problem}[\OneEnd]\label{prob:one-end}
Given a supersingular elliptic curve \(E\) defined over \(\F_{p^2}\), find an
endomorphism in \(\End(E)\setminus\Z\).
\end{problem}

For a function \(\lambda\), the bounded problem \(\OneEnd_\lambda\) of
\cite{MW25} further requires
\(\log\deg\alpha\leq\lambda(\log p)\).

\begin{problem}[\EndRing]\label{prob:end-ring}
Given a supersingular elliptic curve \(E\) defined over \(\F_{p^2}\), find
four endomorphisms generating \(\End(E)\) as a \(\Z\)-module.
\end{problem}

\begin{problem}[\Isogeny]\label{prob:isogeny}
Given supersingular elliptic curves \(E\) and \(E'\) defined over
\(\F_{p^2}\), compute an isogeny \(\varphi:E\to E'\).
\end{problem}

Page--Weso{\l}owski and Herl\'edan Le Merdy--Weso{\l}owski give probabilistic
reductions between these three problems~\cite{PW24,MW25}.  The reductions run
in polynomial time and are unconditional, so the new complexity bound may be
proved for \OneEnd\ and then transferred to the other two problems.

\paragraph{Computational model.}
We use a word RAM model with word size \(\Theta(\log p)\), in which field
operations in \(\F_{p^2}\) and memory accesses have unit cost when only the
exponent of \(p\) is at issue.  We retain polynomial factors in \(\log p\) and every
subexponential factor that is displayed.  Isogenies of small prime degree can
be evaluated with V\'elu formulas or modular polynomials, and every table
constructed below has size within the stated complexity bound.  Throughout
the paper, \(\log\) is the natural logarithm and \(\log_2\) has base two.

\subsection{Models, Frobenius, and isogenies to the conjugate}

For a prime \(p>3\) and each supersingular invariant \(j\), fix a model
\(E_j/\F_{p^2}\) on which the \(p^2\)-Frobenius is \([-p]\); this model is
maximal over \(\F_{p^2}\), since
\(\#E_j(\F_{p^2})=(p+1)^2\).  If \(j\notin\{0,1728\}\), this condition selects
one of the two quadratic twists over \(\F_{p^2}\); such invariants are called
\emph{generic}.

As in~\cite{Wes26,Udovenko26}, \(E^{(p)}\) denotes the curve obtained by
applying the \(p\)-power map to every coefficient of \(E\).  Following
Udovenko, write
\[
 \pi_E:E\longrightarrow E^{(p)},
 \qquad \pi_{E^{(p)}}:E^{(p)}\longrightarrow E
\]
for the two relative Frobenius maps.  On the models just fixed, they satisfy
\begin{equation}\label{eq:frobenius-square}
 \pi_{E^{(p)}}\circ\pi_E=[-p],\qquad
 \widehat{\pi_{E^{(p)}}}=-\pi_E.
\end{equation}

\paragraph{Comparison of notation.}
The same maps have different names in the two references: our
\(\varphi:E\to E^{(p)}\) is denoted by \(\phi\) in~\cite{Wes26} and by
\(\psi\) in~\cite{Udovenko26}, while our \(\pi\) is denoted by \(\varphi\)
in~\cite{Wes26} and by \(\pi\) in~\cite{Udovenko26}.  We use \(\omega\) for
the initial walk, as in~\cite{Wes26}; in particular,
\(\pi_{E^{(p)}}\circ\varphi\) is the map written
\(\varphi\circ\phi\) in~\cite{Wes26} and
\(\pi\circ\psi\) in~\cite{Udovenko26}.

For an endomorphism $u$, write
\[
 \Trd(u)=u+\widehat u\in\Z,
 \qquad \Nrd(u)=\deg u.
\]
The degree is a positive definite quadratic form with polar form
\(B(u,v)=\Trd(\widehat u\circ v)\), and hence
\begin{equation}\label{eq:CS-degree}
 |\Trd(\widehat u\circ v)|\leq2\sqrt{\deg u\deg v}.
\end{equation}

\begin{definition}[Chenu--Smith structure]
Let \(d\) be prime to \(p\) and let \(\varepsilon\in\{\pm1\}\).  A pair
\((E,\varphi)\), where
\(\varphi:E\to E^{(p)}\) is a separable cyclic isogeny of degree \(d\), is a
\emph{$(d,\varepsilon)$-structure} if
\begin{equation}\label{eq:structure-sign}
 \widehat\varphi=\varepsilon\varphi^{(p)}.
\end{equation}
All subsequent counts use \(\varepsilon=+1\); on a generic maximal model, the
maps \(\varphi\) and \(-\varphi\) have the same kernel, which the counts
include once.
\end{definition}

Let $r_d(E)$ be the number of cyclic kernels supporting a
$(d,+1)$-structure on $E$.
For squarefree $d,e$ coprime to $p$, define the ordered collision count
\begin{equation}\label{eq:collision-definition}
 C^+_{d,e}(p)=
 \sum_E r_d(E)r_e(E)
 -\mathbf 1_{d=e}\sum_Er_d(E).
\end{equation}
The sum runs over the generic maximal models fixed above.  When \(d=e\), the
subtracted term removes the pair formed from one kernel twice, while two
distinct kernels of the same degree still contribute.

\subsection{Signs of small conjugate isogenies}

The sign used in the counts is forced for every sufficiently small isogeny to
the conjugate.  Indeed, the completion of \(\End(E)\) at \(p\) is the maximal
order in the quaternion division algebra over \(\Q_p\), and every inseparable
endomorphism belongs to its maximal two-sided ideal.  Its reduced trace
therefore lies in \(p\Z_p\); see~\cite[Chapter~13]{Voight21}.

\begin{lemma}\label{lem:trace-divisible}
If $u\in\End(E)$ is inseparable, then $p\mid\Trd(u)$.
\end{lemma}

\begin{proposition}[Sign of small conjugate isogenies]
\label{prop:automatic}
Let $\varphi:E\to E^{(p)}$ be separable and cyclic of degree $d$.  If $4d<p$,
then $(E,\varphi)$ is a $(d,+1)$-structure; equivalently,
$\widehat\varphi=\varphi^{(p)}$.
\end{proposition}

\begin{proof}
The map \(\alpha=\pi_{E^{(p)}}\circ\varphi\) is inseparable of degree \(pd\),
so Lemma~\ref{lem:trace-divisible} gives \(p\mid\Trd(\alpha)\), while
Inequality~\eqref{eq:CS-degree} gives
\[
 |\Trd(\alpha)|\leq2\sqrt{pd}<p.
\]
The only multiple of \(p\) in this interval is zero, so
\(\Trd(\alpha)=0\) and \(\widehat\alpha=-\alpha\).
Equation~\eqref{eq:frobenius-square} and the naturality of Frobenius now give
\(\widehat\varphi\circ\pi_E=\varphi^{(p)}\circ\pi_E\); the common right
factor \(\pi_E\) is surjective, so
\(\widehat\varphi=\varphi^{(p)}\).
\end{proof}

\subsection{Orientations and ideal isogenies}

The collision proof uses the following part of the correspondence between
isogenies and ideals.

Let \(K\) be an imaginary quadratic field, and let \(\OO\subset K\) be an
order.  An \(\OO\)-orientation of \(E\) is an embedding
\(\iota:\OO\hookrightarrow\End(E)\).  The orientation is \emph{primitive},
also called \emph{optimal}, if
\[
 \iota(K)\cap\End(E)=\iota(\OO).
\]
An isomorphism from \((E,\iota)\) to \((E',\iota')\) is an isomorphism
\(f:E\to E'\) such that
\(f\circ\iota(a)=\iota'(a)\circ f\) for every \(a\in\OO\).

If \(\idealA\) is an invertible integral \(\OO\)-ideal prime to \(p\), its
kernel subgroup is
\[
 E[\idealA]
 =
 \{P:\iota(a)P=0\text{ for every }a\in\idealA\}.
\]
The quotient by this subgroup has degree \(N(\idealA)\), and its target
inherits an \(\OO\)-orientation.  Because principal ideals preserve the
oriented isomorphism class, this construction defines an action of
\(\Pic(\OO)\).  Onuki's reduction theorem shows that the primitive
orientations needed below form at most two free torsors under this
group~\cite{Onuki20}, as used in \cref{sec:class-square}.

An oriented isogeny is \emph{horizontal} if its source and target have the
same order.  A horizontal cyclic isogeny whose degree is prime to the
conductor is represented by an invertible ideal.  Below, this correspondence
is applied in the forward direction to isogenies that already exist.

\subsection{The quadratic order attached to a pair}

Each pair counted by \eqref{eq:collision-definition} determines a quadratic
order as follows.  Let
\(\varphi,\psi:E\to E^{(p)}\) define \((d,+1)\)- and
\((e,+1)\)-structures, where \(d\) and \(e\) are squarefree.  Put
\begin{equation}\label{eq:gamma}
 \gamma=\widehat\varphi\circ\psi,\qquad
 t=\Trd(\gamma),\qquad
 \Delta=4de-t^2.
\end{equation}
Since the reduced norm of \(\gamma\) is \(de\),
\begin{equation}\label{eq:gamma-poly}
 \gamma^2-t\gamma+de=0.
\end{equation}

\begin{lemma}\label{lem:non-scalar-pair}
The subtracted pair in~\eqref{eq:collision-definition} is precisely the case
in which \(\gamma\) is scalar.  Every nonsubtracted pair has \(\Delta>0\).
\end{lemma}

\begin{proof}
If \(\gamma=[a]\), then \(a^2=de\), and the squarefreeness of \(d\) and \(e\)
implies \(d=e\) and \(a=\pm d\).  The identity
\(\varphi\circ\gamma=[d]\circ\psi\) then gives
\(\psi=\pm\varphi\), so the kernels coincide; conversely, equal kernels give
this scalar case.  For non-scalar \(\gamma\), the strict form of
\eqref{eq:CS-degree} gives \(t^2<4de\).
\end{proof}

For a nonsubtracted pair, let
\[
 K=\Q(\theta),\qquad \theta^2-t\theta+de=0,
 \qquad \iota(\theta)=\gamma,
\]
and define the order selected by this embedding:
\begin{equation}\label{eq:optimal-order}
 \OO=\iota^{-1}\bigl(\End(E)\cap\iota(K)\bigr).
\end{equation}
This order is primitive by definition, and if
\(h=[\OO:\Z[\theta]]\) and \(D=\disc(\OO)<0\), then
\begin{equation}\label{eq:disc-relation}
 -\Delta=h^2D.
\end{equation}
The criterion for the existence of a primitive orientation implies that
\(p\) does not split in \(K\) and that
\(p\nmid\cond(\OO)\)~\cite[Proposition~3.2]{Onuki20}.

Let \(c\) be the nontrivial automorphism of \(K/\Q\).  The two identities with
sign \(+1\) give
\begin{equation}\label{eq:intertwining}
 \varphi\circ\gamma=[d]\circ\psi
 =\widehat\gamma^{(p)}\circ\varphi.
\end{equation}
Thus \(\varphi\) is an oriented isogeny from \(x=(E,\iota)\) to
\begin{equation}\label{eq:J}
 Jx=(E^{(p)},\iota^{(p)}\circ c).
\end{equation}

\begin{lemma}[Order preservation for squarefree degrees]\label{lem:horizontal}
The oriented isogeny $\varphi:x\to Jx$ is horizontal, is represented by an
invertible integral $\OO$-ideal $\idealA$ of norm $d$, and satisfies
\begin{equation}\label{eq:conductor-coprime}
 \gcd(de,\cond(\OO))=1.
\end{equation}
\end{lemma}

\begin{proof}
Factor \(\varphi\) into steps of prime degree, and transport the orientation
through the factorization.  A step of degree \(\ell\) can change the order
only at \(\ell\), because after \(\ell\) is inverted, conjugation by that step
identifies the two endomorphism rings.  Since \(d\) is squarefree, there is at
most one step of degree \(\ell\), while every other step preserves the
localization at \(\ell\).  The orders at the start and at the end are both
\(\OO\), so the step of degree \(\ell\) must preserve that localization as
well.  Applying this argument to each \(\ell\mid d\) shows that every step is horizontal,
and \(\varphi\) is represented by an invertible ideal of norm
\(d\)~\cite[Section~3]{Onuki20}.  Applying the same argument after exchanging
\(\varphi\) and \(\psi\) proves~\eqref{eq:conductor-coprime}.
\end{proof}

\section{The collision theorem}
\label{sec:collisions}

A pair \((\varphi,\psi)\) as in \cref{sec:prelim} determines the order
\(\OO\), the element \(\theta\), and an ideal \(\idealA\) of norm \(d\).  The
existence of \(\psi\) restricts the possible ideals \(\idealA\), and the
endpoint condition becomes a square equation in \(\Pic(\OO)\), whose fibers
are bounded by genus theory.  An elementary divisor estimate then sums these
bounds over the trace and the order, yielding an upper bound for the number
of curves on which two degrees occur.

\subsection{The compatibility condition}

Let \(\idealA\) be the invertible \(\OO\)-ideal of norm \(d\) associated with
\(\varphi\) by Lemma~\ref{lem:horizontal}.  The relation
\(\idealA\overline{\idealA}=d\OO\) identifies the subgroup
\(E[\idealA]\subset E[d]\) with \(\overline{\idealA}/d\OO\), after a choice of
module coordinates.

\begin{lemma}[Restriction from the second isogeny]
\label{lem:compatibility}
For every collision pair,
\begin{equation}\label{eq:ideal-compatibility}
 (\theta)\subseteq\overline{\idealA}.
\end{equation}
\end{lemma}

\begin{proof}
For \(\ell\mid d\), Equation~\eqref{eq:conductor-coprime} shows that
\(\OO\otimes\Z_\ell\) is the maximal order of \(K\otimes\Q_\ell\).
The Tate module \(T_\ell(E)\) is free of rank one over this order; when the
local algebra is a field, its order is a discrete valuation ring, and when it
is split, its two idempotents give two summands of rank one over \(\Z_\ell\).
After choosing an identification, we have
\[
 E[d]\simeq\OO/d\OO
\]
as $\OO$-modules, and $E[\idealA]$ corresponds to
$\overline{\idealA}/d\OO$.

The identity~\eqref{eq:intertwining} reads
\(\varphi\circ\iota(\theta)=[d]\circ\psi\), whose right side kills
\(E[d]\); hence
\(\iota(\theta)E[d]\subseteq\ker\varphi=E[\idealA]\).  In the chosen
coordinates, this inclusion becomes
\[
 \theta(\OO/d\OO)\subseteq\overline{\idealA}/d\OO.
\]
Because \(d\OO=\idealA\overline{\idealA}\) is contained in
\(\overline{\idealA}\), the last inclusion is equivalent to
\(\theta\OO\subseteq\overline{\idealA}\), which is the required ideal
condition.
\end{proof}

The ideals allowed by this condition can be counted in the coordinates
\begin{equation}\label{eq:omega-coordinates}
 \delta\equiv D\pmod 2,\qquad
 \omega=\frac{\delta+\sqrt D}{2},\qquad
 \theta=a_0+h\omega,\qquad
 a_0=\frac{t-h\delta}{2}.
\end{equation}
The traces of \(\theta\) and \(h\omega\) have the same parity, so
\(a_0\in\Z\).

\begin{lemma}[Compatible ideals of prescribed norm]
\label{lem:ideal-count}
For a fixed triple $(t,\OO,h)$ satisfying~\eqref{eq:disc-relation}, the number
of invertible integral ideals $\idealA$ of norm $d$ satisfying
\eqref{eq:ideal-compatibility} is
\begin{equation}\label{eq:ideal-count}
 I_{D,h}(d,e,t)=
 \prod_{\ell\mid\gcd(d,h)}
 \left(1+\left(\frac{D}{\ell}\right)\right),
\end{equation}
where the symbol at $2$ is the Kronecker symbol.  In particular,
\begin{equation}\label{eq:ideal-count-upper}
 I_{D,h}(d,e,t)
 \leq2^{\omega(\gcd(d,e,t))}.
\end{equation}
\end{lemma}

\begin{proof}
Because the integer \(d\) is squarefree and prime to the conductor, the ideal
can be chosen separately at each \(\ell\mid d\).  Ideals of norm
\(\ell\) correspond to roots modulo \(\ell\) of the monic polynomial of
\(\omega\), and the condition \(\theta\in\overline{\idealA}\) imposes one
linear congruence on that root.

If \(\ell\nmid h\), the congruence selects the residue \(-a_0/h\).  Since
\(N(\theta)=de\equiv0\pmod\ell\), this residue is a root of the polynomial of
\(\omega\) and gives one local ideal.

If \(\ell\mid h\), the norm identity gives \(\ell\mid a_0\), so the linear
condition imposes no restriction and the number of roots is
\(1+(D/\ell)\): two if \(\ell\) splits, one if it ramifies, and zero if it is
inert.  The same statement holds at \(2\) with the Kronecker symbol, and
multiplication of the local counts proves~\eqref{eq:ideal-count}.

A local factor equals two only when \(\ell\mid d,h\) and \(\ell\) splits in
\(\OO\).  The relation \(h^2D=t^2-4de\), together with squarefreeness, then
implies \(\ell\mid e,t\).  For odd \(\ell\), this follows by reducing the
relation modulo \(\ell^2\); for \(\ell=2\), it follows from
\(D\equiv1\pmod8\) and the norm formula in
\eqref{eq:omega-coordinates}.  Hence every prime that contributes a factor
two divides \(\gcd(d,e,t)\), proving \eqref{eq:ideal-count-upper}.
\end{proof}

\subsection{The square equation in the class group}
\label{sec:class-square}

For \(G=\Pic(\OO)\), Onuki proves that the action on the reduction orbit of a
primitive \(\OO\)-oriented curve is free and transitive.  Every primitive
orientation needed here belongs either to that orbit or to its Frobenius
image~\cite[Proposition~3.3 and Theorem~3.4]{Onuki20}, so at most two
\(G\)-torsors occur.

Conjugation inverts an ideal class, whereas Frobenius commutes with the ideal
action; hence the operation \(J\) from~\eqref{eq:J} satisfies
\begin{equation}\label{eq:J-action}
 J(\idealB*x)=\idealB^{-1}*Jx.
\end{equation}
Consider one of the torsors on which the endpoint condition has a solution.
Choose \(x_0\) in that torsor with \(Jx_0=c_0*x_0\), and write every other
point as \(x=\idealB*x_0\).  The isogeny represented by \(\idealA\) ends at
\(Jx\) precisely when
\[
 \idealA*x=Jx,
\]
or equivalently
\begin{equation}\label{eq:square-equation}
 \idealB^2=[\idealA]^{-1}c_0.
\end{equation}
Thus the endpoint condition is a square equation in \(G\), and each nonempty
fiber of the square map has cardinality \(|G[2]|\).  The two possible torsors
therefore contribute at most \(2|G[2]|\) oriented curves for each allowed
ideal.

The passage from oriented maps to the kernels counted by
\(C^+_{d,e}(p)\) requires accounting for signs.  Each kernel determines its
isogeny up to sign, and among the
four choices \((\pm\varphi,\pm\psi)\), changing both signs leaves
\(\gamma\) unchanged, whereas changing one sign replaces \(\gamma\) by
\(-\gamma\).  A pair of kernels therefore gives two oriented data:
\[
 \gamma=\widehat\varphi\circ\psi
 \quad\text{and}\quad -\gamma.
\]
These data are distinct even when \(t=0\), because the automorphisms of a
generic curve are the central elements \(\pm1\), whose conjugation fixes every
endomorphism; a non-scalar endomorphism and its negative therefore define
distinct data.

Conversely, the oriented datum fixes \(\gamma\), and \(\idealA\) fixes the
kernel of \(\varphi\).  For data obtained from an actual pair, the identity
\begin{equation}\label{eq:recover-second-kernel}
 \psi=\frac{\varphi\circ\gamma}{d}
\end{equation}
is integral and fixes the kernel of \(\psi\).  Changing both signs changes
neither kernel, so the two data obtained from each pair cancel the factor two
from the two torsors and give
\begin{equation}\label{eq:geometric-injection}
 C^+_{d,e}(p)
 \leq
 \sum_{t^2<4de}\ \sum_{\OO\supseteq\Z[\theta]}
 I_{D,h}(d,e,t)\,|\Pic(\OO)[2]|.
\end{equation}

Equation~\eqref{eq:geometric-injection} injects the geometric pairs into the
arithmetic data on the right, so counting all such data gives a valid upper
bound even when some terms are not realized by isogenies.

\subsection{Sum over quadratic orders}

The sum over \(\OO\) is controlled by a bound for the \(2\)-torsion in its
class group.  The following constant applies to nonmaximal orders and also
covers the prime \(2\).

\begin{lemma}[Quadratic $2$-torsion]
\label{lem:two-torsion}
If $\OO$ is an imaginary quadratic order of discriminant $D<0$, then
\begin{equation}\label{eq:two-torsion}
 |\Pic(\OO)[2]|\leq13\,2^{\omega(|D|)}.
\end{equation}
\end{lemma}

\begin{proof}
An element of order at most two is represented by a primitive, reduced,
ambiguous binary quadratic form.  Such a form satisfies one of
\(b=0\), \(|b|=a\), or \(a=c\).  In these three cases the discriminant
factors as
\[
 |D|=4ac,\qquad |D|=a(4c-a),\qquad
 |D|=(2a-b)(2a+b).
\]
Primitivity makes the two factors coprime outside \(2\).  Each odd prime
power must therefore occur wholly in one factor.  When \(b=0\), the power of
\(2\) has at most two allocations, and each of the other two cases has at most
four.  Once the prime powers have been assigned, the reduction inequalities
leave at most one form for each remaining choice.  The total is at most
\((2+4+4)2^{\omega(|D|)}\), which is bounded by the right side of
\eqref{eq:two-torsion}, as in the usual proof through ambiguous
forms~\cite[Chapter~3]{Cox13}.
\end{proof}

The orders containing \(\Z[\theta]\) have discriminants
\(-\Delta/h^2\), and the identity
\begin{equation}\label{eq:divisor-identity}
 \sum_{h^2\mid n}2^{\omega(n/h^2)}=\tau(n)
\end{equation}
holds because both sides are multiplicative and equal \(a+1\) when
\(n=\ell^a\).  Applying it to~\eqref{eq:geometric-injection}, together with
Lemmas~\ref{lem:ideal-count} and~\ref{lem:two-torsion}, gives
\begin{equation}\label{eq:pre-arithmetic-bound}
 C^+_{d,e}(p)
 \leq13\sum_{t^2<4de}
 2^{\omega(\gcd(d,e,t))}\tau(4de-t^2).
\end{equation}

\subsection{An average divisor bound}

The remaining sum involves only elementary arithmetic.  For \(m\geq1\), put
\[
 S(m)=\sum_{t^2<4m}\tau(4m-t^2),
 \qquad
 \PP(X)=\prod_{\substack{\ell\leq X\\ \ell\ {\rm prime}}}
 \frac{\ell+1}{\ell-1}.
\]

\begin{lemma}\label{lem:divisor-average}
For every real $Z\geq2\sqrt m$,
\begin{equation}\label{eq:divisor-average}
 S(m)\leq6Z\PP(Z).
\end{equation}
\end{lemma}

\begin{proof}
Let
\(\rho_A(a)=\#\{u\bmod a:u^2\equiv A\pmod a\}\).  Applying
\(\tau(n)\leq2\sum_{a\leq\sqrt n,\,a\mid n}1\) and summing first over the
possible divisors \(a\), the number of representatives of each residue class
in \(|t|<Z\) gives
\begin{equation}\label{eq:root-sum}
 S(m)\leq6Z\sum_{a\leq Z}\frac{\rho_{4m}(a)}a.
\end{equation}
For every prime $\ell$ and integer $A$,
\begin{equation}\label{eq:local-root-sum}
 \sum_{j\geq1}\frac{\rho_A(\ell^j)}{\ell^j}
 \leq\frac2{\ell-1}.
\end{equation}
For the local bound, write \(v=v_\ell(A)\).  At level \(j\leq v\), the zero
congruence gives \(\ell^{\lfloor j/2\rfloor}\) roots.  An odd \(v\) gives no
roots above level \(v\); if \(v=2a\), division by \(\ell^{2a}\) leaves a
congruence asking for the square root of a unit, each root of which has
\(\ell^a\) lifts.

For odd \(\ell\), a unit has at most two square roots at each level, and the
two geometric sums are bounded by \(2/(\ell-1)\).  For \(\ell=2\), a unit has
at most \(1,2,4\) roots modulo \(2,4,2^j\); the part arising from zero is
\(2(1-2^{-a})\), and the remaining tail is at most \(2^{1-a}\), with total at
most \(2\), proving~\eqref{eq:local-root-sum}.

The multiplicativity of \(\rho_A\) gives
\[
 \sum_{a\leq Z}\frac{\rho_{4m}(a)}a
 \leq\prod_{\ell\leq Z}\left(1+\frac2{\ell-1}\right)
 =\PP(Z).
\]
Substituting this estimate in~\eqref{eq:root-sum} proves the lemma.
\end{proof}

\begin{theorem}[Finite collision bound]\label{thm:collision}
Let $p>3$ and let $d,e\geq2$ be squarefree and coprime to $p$.  Set
\[
 X=2\ceil{\sqrt{de}}.
\]
Then
\begin{equation}\label{eq:finite-collision}
 \boxed{
 C^+_{d,e}(p)
 \leq78X\PP(X)
 \prod_{\ell\mid\gcd(d,e)}\left(1+\frac3\ell\right).}
\end{equation}
In particular,
\begin{equation}\label{eq:asymptotic-collision}
 C^+_{d,e}(p)
 \ll \sqrt{de}\,\log^2(2de)
 \prod_{\ell\mid\gcd(d,e)}\left(1+\frac3\ell\right),
\end{equation}
with an absolute implied constant, uniformly in $p,d,e$.
\end{theorem}

\begin{proof}
Set \(m=de\) and \(g=\gcd(d,e)\); since \(g\) is squarefree,
\[
 2^{\omega(\gcd(g,t))}
 =\sum_{\substack{r\mid g\\r\mid t}}1.
\]
For each such \(r\), the relations \(r^2\mid m\), \(t=ru\), and
\(\tau(r^2n)\leq3^{\omega(r)}\tau(n)\) show that the right side of
\eqref{eq:pre-arithmetic-bound} is then at most
\[
 13\sum_{r\mid g}3^{\omega(r)}S(m/r^2).
\]
Lemma~\ref{lem:divisor-average}, with \(Z=X/r\), and the inequality
\(\PP(X/r)\leq\PP(X)\) give
\[
 78X\PP(X)\sum_{r\mid g}\frac{3^{\omega(r)}}r
 =78X\PP(X)\prod_{\ell\mid g}\left(1+\frac3\ell\right),
\]
which is~\eqref{eq:finite-collision}.

Since
\((\ell+1)/(\ell-1)=(1-\ell^{-2})(1-\ell^{-1})^{-2}\), the Mertens product
estimate of Rosser and Schoenfeld~\cite{RS62} gives
\(\PP(X)=O(\log^2X)\), with the uniform bound
\begin{equation}\label{eq:explicit-P}
 \PP(X)\leq2^{40}(1+\log X)^2\qquad(X\geq2).
\end{equation}
Since \(X\leq3\sqrt{de}\), this estimate also proves
\eqref{eq:asymptotic-collision}.
\end{proof}

\begin{corollary}[Uniform form for a bounded family]
\label{cor:COL50}
Let $n=\floor{\log_2p}+1$.  If $d,e\leq B$, $4B<p$, and the hypotheses of
Theorem~\ref{thm:collision} hold, then
\begin{equation}\label{eq:COL50}
 C^+_{d,e}(p)
 \leq2^{50}n^2\sqrt{de}
 \prod_{\ell\mid\gcd(d,e)}\left(1+\frac3\ell\right).
\end{equation}
\end{corollary}

\begin{proof}
Here \(X\leq2B<p/2\), so \(1+\log X\leq n\).  Combining
\eqref{eq:finite-collision} with~\eqref{eq:explicit-P}, and then using
\(X\leq3\sqrt{de}\) and \(234\cdot2^{40}<2^{50}\), proves the claim.
\end{proof}

\section{Many curves with smooth isogenies to their conjugates}
\label{sec:density}

The degree family used by the algorithm is fixed deterministically before the
sampling in \cref{sec:algorithm}.  The Chenu--Smith count gives the total
number of associated isogenies, and the collision theorem shows that these
isogenies occur on many curves.  Denote the bit length of \(p\) by
\[
 n=\floor{\log_2p}+1.
\]

\subsection{Isogeny count for one degree}

\begin{proposition}[Count for one degree]\label{prop:CS-mass}
For every squarefree \(d\geq2\) prime to \(p\),
\begin{equation}\label{eq:first-moment-raw}
 M_d:=\sum_Er_d(E)
 \geq\frac12h_{K_d}-5\Psi(d),
\end{equation}
where \(K_d=\Q(\sqrt{-pd})\), \(D_{K_d}\) is its fundamental discriminant,
\(h_{K_d}=h(D_{K_d})\) is its class number, and
$\Psi(d)=d\prod_{\ell\mid d}(1+1/\ell)$.
\end{proposition}

Here \(\Psi\) denotes the Dedekind psi function, as
in~\cite{Udovenko26}; Weso{\l}owski uses the distinct notation \(\Psi(X,B)\)
for the number of smooth integers in an interval~\cite{Wes26}.

\begin{proof}
In the notation of~\cite[Theorem~3]{Udovenko26}, Chenu and Smith count the
isomorphism classes of \((d,+1)\)-structures by
\[
 \alpha_d=
 \begin{cases}
  2h_{K_d},&-dp\equiv1\pmod 8,\\
  4h_{K_d},&-dp\equiv5\pmod 8,\\
  h_{K_d},&\text{otherwise}.
 \end{cases}
\]
The count agrees with~\cite[Corollary~4.15]{CS22} for sign \(+1\), and in every case
\(\alpha_d\geq h_{K_d}\).  Since the Frobenius trace of a
\((d,\varepsilon)\)-structure is \(-2\varepsilon p\), our choice of maximal
model selects \(\varepsilon=+1\) and fixes the twist.

To pass from signed maps to the kernels counted by \(r_d(E)\), observe that
there are
\[
 \Psi(d)=\prod_{\ell\mid d}(\ell+1)
\]
cyclic subgroups of order \(d\) in \(E[d]\).  After a kernel is fixed, the
quotient map and an identification of its target with \(E^{(p)}\) are
determined up to an automorphism of \(E^{(p)}\).

The automorphism groups at \(j=0\) and \(j=1728\) have orders \(6\) and \(4\),
so the structures above these two invariants contribute at most
\(10\Psi(d)\) terms.  On every generic maximal model, the automorphism group
is \(\{\pm1\}\), and the maps \(\varphi\) and \(-\varphi\) give one kernel.
Removing the two special invariants and dividing by two yields
\[
 M_d\geq\frac{\alpha_d-10\Psi(d)}2
 \geq\frac12h_{K_d}-5\Psi(d),
\]
which proves~\eqref{eq:first-moment-raw}.
\end{proof}

For the degrees used below, \(\Psi(d)=d^{1+o(1)}\), whereas the class number
term has order \(\sqrt{pd}/\log p\); every estimate retains the lower-order
subtraction.

\begin{proposition}[Class number bound for the family]\label{prop:first-moment}
Let $p\geq2^{17}$ be prime, and let $\DD$ be any family of squarefree degrees
$d\geq2$ satisfying $p\nmid d$ and $4d<p$.  Apart from at most one degree
$d_*\in\DD$, every $d\in\DD$ satisfies
\begin{equation}\label{eq:first-moment}
 M_d\geq\frac{\sqrt{pd}}{80n}-5\Psi(d).
\end{equation}
\end{proposition}

\begin{proof}
Apply Tatuzawa's theorem~\cite{Tatuzawa51} with
\(\varepsilon=1/\log p\) for every degree in \(\DD\).  The fundamental
discriminant satisfies
\[
 pd\leq|D_{K_d}|\leq4pd<p^2.
\]
For \(p\geq2^{17}\), this discriminant also satisfies
\(|D_{K_d}|\geq\max(e^{1/\varepsilon},e^{11.2})\).  Tatuzawa's theorem
therefore gives, with at most one primitive real character excluded,
\[
 L(1,\chi_{D_{K_d}})>
 0.655\,\varepsilon |D_{K_d}|^{-\varepsilon}.
\]
The class number formula for imaginary quadratic fields and the inequality
\(|D_{K_d}|^{-1/\log p}>e^{-2}\) imply
\[
 \frac12h_{K_d}
 >\frac{0.655}{2\pi e^2\log p}\sqrt{pd}
 >\frac{\sqrt{pd}}{80n}.
\]
For the last inequality, use
\(\log p<n\log2\) and \(2\pi e^2\log2<80\cdot0.655\).

Distinct squarefree degrees prime to \(p\) give distinct quadratic fields, so
the exceptional character can affect at most one degree in \(\DD\).
Substitution of the class number bound in \eqref{eq:first-moment-raw} proves
the proposition.
\end{proof}

\subsection{Distinct curves}

The preceding sum counts pairs \((E,d)\), whereas the density estimate
requires the number of distinct curves that occur among them.  The following
lemma combines the total count with the collision bounds: its first inequality
counts points in the support, and its second counts points with two distinct
labels.

\begin{lemma}\label{lem:moments}
Let $\mathcal X$ be finite, $|\mathcal X|\leq N$, and let
$r_i:\mathcal X\to\Z_{\geq0}$.  Put
\[
 R(x)=\sum_i r_i(x),\qquad S=\sum_xR(x),
\]
and
\[
 C_{ij}=\sum_xr_i(x)r_j(x)-\mathbf1_{i=j}\sum_xr_i(x).
\]
Suppose $S\geq S_0\geq0$ and $C_{ij}\leq U_{ij}$, where $U$ is symmetric and
nonnegative.  If $K=\sum_{i,j}U_{ij}$, then
\begin{equation}\label{eq:coverage-moment}
 \#\{x:R(x)>0\}\geq\frac{S_0^2}{S_0+K}.
\end{equation}
Moreover, let
\[
 V(x)=\sum_i\mathbf1_{r_i(x)>0},\quad
 V_0=\max\left\{0,S_0-\frac12\sum_iU_{ii}\right\},\quad
 K_{\ne}=\sum_{i\ne j}U_{ij}.
\]
If $V_0>N$, then
\begin{equation}\label{eq:two-label-moment}
 \#\{x:V(x)\geq2\}
 \geq\frac{(V_0-N)^2}{V_0+K_{\ne}}.
\end{equation}
\end{lemma}

\begin{proof}
By the definition of \(C_{ij}\),
\[
 \sum_xR(x)^2=S+\sum_{i,j}C_{ij}\leq S+K.
\]
Cauchy--Schwarz on the support of \(R\), together with the monotonicity of
\(s^2/(s+K)\), gives the first bound.

For every integer $r\geq0$,
$\mathbf1_{r>0}\geq r-r(r-1)/2$, and hence
$\sum_xV(x)\geq V_0$.  Moreover,
\[
 \sum_xV(x)^2
 \leq\sum_xV(x)+K_{\ne}.
\]
The points with \(V(x)\leq1\) contribute at most \(N\) to
\(\sum_xV(x)\); applying Cauchy--Schwarz to the remaining points and using
the same monotonicity argument gives~\eqref{eq:two-label-moment}.
\end{proof}

The argument is deterministic and is based on the displayed count and
collision bounds.

\subsection{The degree family}

The degrees searched by the algorithm contain the same number of prime
factors, all drawn from a short interval.  The interval is chosen so that the
degrees are large enough to reach many curves and their factors remain small
enough for enumeration.

For $n\geq2^{16}$ let
\begin{equation}\label{eq:parameters}
 h=\ceil{\log_2n},\qquad
 s=\min\{a\in2\Z:a^2\geq nh\},\qquad
 k=\ceil{\frac{n}{3s}}+4.
\end{equation}
Let $q_1<\cdots<q_m$ be all primes in
\begin{equation}\label{eq:prime-interval}
 2^s<q\leq2^{s+4},
\end{equation}
and define
\begin{equation}\label{eq:degree-family}
 \DD_p=
 \left\{\prod_{i\in I}q_i:
 I\subseteq\{1,\ldots,m\},\ |I|=k\right\}.
\end{equation}
Every member of $\DD_p$ is squarefree, and all its prime factors are at most
$2^{s+4}$.

The elementary prime estimates of Rosser and Schoenfeld~\cite{RS62} imply
\begin{equation}\label{eq:prime-count}
 m\geq\frac{4\cdot2^s}{s}
\end{equation}
in our range.  This weak bound can also be recovered directly from central
binomial coefficients and Chebyshev's function.

For each $q_i$, set
\[
 w_i=\floor{\sqrt{q_i}},\qquad
 v_i=\ceil{\sqrt{q_i}},\qquad
 b_i=v_i^2+\ceil{\frac{3v_i^2}{q_i}}.
\]
Let $e_k$ denote the elementary symmetric polynomial of degree $k$, and put
\begin{align}
 A&=e_k(w_1,\ldots,w_m),
 &W&=\prod_{i=m-k+1}^mw_i,\label{eq:AW}\\
 Z&=e_k(q_1+1,\ldots,q_m+1),
 &J&=e_k(q_1+3,\ldots,q_m+3),\label{eq:ZJ}\\
 H&=[X^kY^k]\prod_{i=1}^m
 (1+v_iX+v_iY+b_iXY),
 &H_{\rm diag}&=e_k(b_1,\ldots,b_m).
 \label{eq:H}
\end{align}

These five quantities collect the terms in the lower bound and in the
collision bounds:
\begin{center}
\begin{tabular}{cl}
\toprule
Quantity & Use in the calculation\\
\midrule
\(A\) & lower bound for \(\sum_{d\in\DD_p}\sqrt d\)\\
\(W\) & largest term that may be lost to Tatuzawa's exception\\
\(Z\) & sum of the penalties involving the Dedekind psi function\\
\(J\) & upper bound for collisions that use the same degree\\
\(H\) & upper bound for all ordered pairs of colliding degrees\\
\bottomrule
\end{tabular}
\end{center}

Expanding the four terms in each factor of the definition of \(H\) gives
\begin{equation}\label{eq:H-majorizes}
 H\geq
 \sum_{d,e\in\DD_p}\sqrt{de}
 \prod_{\ell\mid\gcd(d,e)}\left(1+\frac3\ell\right),
\end{equation}
and \(H-H_{\rm diag}\) is an upper bound for the same sum over \(d\ne e\).
An index used only by \(d\) or only by \(e\) contributes
\(v_i\geq\sqrt{q_i}\), whereas an index common to both contributes
\[
 b_i\geq q_i+3=q_i\left(1+\frac3{q_i}\right).
\]
For equal degrees, the required upper bound is \(J\), since
\(d\prod_{\ell\mid d}(1+3/\ell)=\prod_{\ell\mid d}(\ell+3)\).

The parameters satisfy the following bounds, whose derivation also gives a
uniform choice of the constants hidden in the \(o(1)\) term:
\begin{align}
 s&=O(\sqrt{n\log n}),&
 k&=O(\sqrt{n/\log n}),\label{eq:sk-size}\\
 B:=\max_{d\in\DD_p}d
 &\leq2^{n/3+6s}<2^{n/2},&4B&<p,\label{eq:B-size}\\
 A&\geq2^{n/2+5s},&W&\leq A/4,\label{eq:A-size}\\
 H&\leq8A^2,&Z,J&\leq2\sqrt B\,A.
\label{eq:rounding-bounds}
\end{align}
The inequalities $\sqrt{nh}\leq s<\sqrt{nh}+2$ and
$n\geq2048h$ hold in the stated range.  Hence $s\geq32h$, $40s\leq n$, and
$k\leq n/(3s)+5$.  These inequalities give
$k(s+4)\leq n/3+6s$ and~\eqref{eq:B-size}.  From
\eqref{eq:prime-count},
$m/k\geq2^{s-h}$, while $w_i\geq2^{s/2}$; therefore
\[
 A\geq\binom mk2^{ks/2}
 \geq2^{k(3s/2-h)}
 \geq2^{n/2+5s}.
\]
The upper bound
\[
 W\leq2^{k(s/2+2)}.
\]
implies \(A/W\geq2^{k(s-h-2)}\geq4\), which proves \(W\leq A/4\).  Finally,
\[
 \frac{v_i}{w_i}\leq1+2^{-s/2},
 \qquad
 \frac{b_i}{v_i^2}\leq1+4\cdot2^{-s}.
\]
Comparing a term of \(H\) with the corresponding ordered pair of terms in
\(A^2\) introduces at most \(2k\) factors of the first type and \(k\) factors
of the second type.  Since
\(k\leq n\leq2^h\leq2^{s/32}\), the logarithm of their product is at most
\[
 2k\,2^{-s/2}+4k\,2^{-s}<\log 8.
\]
Consequently \(H\leq8A^2\).  Similarly, for \(c\in\{1,3\}\) and any \(k\)-set
\(I\),
\[
 \prod_{i\in I}(q_i+c)
 \leq
 2\sqrt{\prod_{i\in I}q_i}\prod_{i\in I}w_i
 \leq2\sqrt B\prod_{i\in I}w_i.
\]
The factor \(2\) in this inequality follows from the bound
\[
 \frac{q_i+c}{\sqrt{q_i}\,w_i}
 \leq\frac{1+3/q_i}{1-q_i^{-1/2}}
\]
and
\[
 k\left(3\cdot2^{-s}
   +\frac{2^{-s/2}}{1-2^{-s/2}}\right)<\log2,
\]
where the last inequality follows from \(s\geq32h\) and \(k\leq2^h\).
Summing over all sets \(I\) proves \(Z,J\leq2\sqrt B\,A\), and hence
\eqref{eq:rounding-bounds}.

\begin{theorem}[Density of smooth isogenies to conjugates]
\label{thm:density}
Let $p>3$ be prime of bit length $n\geq2^{16}$.  For the deterministic family
$\DD_p$ in~\eqref{eq:degree-family}, at least
\begin{equation}\label{eq:density}
 \boxed{\frac{p}{2^{80}n^4}}
\end{equation}
generic supersingular maximal models admit \((d,+1)\)-structures for at least
two distinct degrees $d\in\DD_p$.  Every such degree satisfies
\begin{equation}\label{eq:smooth-degree-size}
d\leq p^{1/3}\exp\bigl(O(\sqrt{\log p\,\log\log p})\bigr),
\end{equation}
and all its prime factors are
$\exp(O(\sqrt{\log p\,\log\log p}))$.
\end{theorem}

\begin{proof}
Remove the possible exceptional degree from
Proposition~\ref{prop:first-moment}, and put
\(\rho=\floor{\sqrt p}\).  The sum of the remaining first moments is at least
\begin{equation}\label{eq:S0}
 S_0=\max\left\{0,
 \frac{\rho(A-W)}{80n}-5Z\right\}.
\end{equation}
To justify the numerator, recall that the primes are ordered increasingly.
The degree with the largest value of \(\sqrt d\) also has the largest product
of the \(w_i\).  Removing any one class number term therefore subtracts at
most \(W\) from the lower bound \(A\).  We retain the penalty involving
\(\Psi\) for the removed degree, thereby obtaining a weaker valid bound.

Let $\Gamma=2^{50}n^2$.  Corollary~\ref{cor:COL50} and
\eqref{eq:H-majorizes} bound all collisions by $\Gamma H$, with contributions
at most $\Gamma J$ from equal degrees and at most
$\Gamma(H-H_{\rm diag})$ from distinct degrees.  Lemma~\ref{lem:moments}
therefore applies with
\begin{equation}\label{eq:V0}
 V_0=\max\{0,S_0-\Gamma J/2\}.
\end{equation}

The comparison between the main term and the error terms begins with
Equation~\eqref{eq:B-size}, which gives
\[
 \sqrt{p/B}\geq2^{n/4-1/2}\geq2^{62}n^3.
\]
Combining this inequality with
\eqref{eq:A-size}--\eqref{eq:rounding-bounds} yields
\begin{equation}\label{eq:V0-lower}
 V_0\geq\frac{\sqrt p\,A}{2048n}\geq2p.
\end{equation}
There are fewer than $p$ generic supersingular invariants, and moreover
$V_0\leq S_0\leq\sqrt pA\leq A^2$ with
\[
 V_0+\Gamma(H-H_{\rm diag})\leq16\Gamma A^2.
\]
Applying~\eqref{eq:two-label-moment} and using
$V_0-p\geq V_0/2$ gives
\[
 \#\{E:\text{at least two degrees}\}
 \geq\frac{V_0^2}{64\Gamma A^2}
 \geq\frac{p}{2^{78}n^4}
 \geq\frac{p}{2^{80}n^4}.
\]
Equations~\eqref{eq:sk-size},~\eqref{eq:B-size}, and
\eqref{eq:prime-interval} give~\eqref{eq:smooth-degree-size} and the stated
bound for the prime factors.
\end{proof}

Because the degree family is fixed before the algorithm samples any curve,
the theorem establishes the factorization property required by the search
from the count and collision estimate above.

\section{Proof of the main result}
\label{sec:algorithm}

The density theorem identifies a set of curves that support the required
isogenies.  The algorithm samples a nearly uniform curve, lists short
isogenies from that curve and its conjugate, matches the two lists, and
transports the resulting endomorphism back to the input.

\subsection{Input model}\label{sec:input-normalization}

The density theorem is stated on maximal models, whereas \OneEnd\ permits an
arbitrary supersingular input \(E_{\mathrm{in}}/\F_{p^2}\); the following
reduction passes between these two descriptions.

When \(j(E_{\mathrm{in}})\in\{0,1728\}\), an automorphism of order \(3\) or
\(4\) already gives a non-scalar endomorphism; on suitable Weierstrass models,
one may use
\((x,y)\mapsto(\zeta_3x,y)\) or \((x,y)\mapsto(-x,\sqrt{-1}\,y)\).
The required roots of unity lie in \(\F_{p^2}\), so the rest of the
construction concerns a generic input.

Let \(E_\star\) be the fixed maximal model with
\(j(E_\star)=j(E_{\mathrm{in}})\).  Because the two curves are quadratic
twists, an isomorphism over the algebraic closure
\begin{equation}\label{eq:input-isomorphism}
 \nu:E_{\mathrm{in}}\longrightarrow E_\star
\end{equation}
is defined over an extension of \(\F_{p^2}\) of degree at most two.

For the computation of the model and the isomorphism, set
\(c=j(E_{\mathrm{in}})/(1728-j(E_{\mathrm{in}}))\) and start from
\[
y^2=x^3+3cx+2c,
\]
whose invariant is \(j(E_{\mathrm{in}})\).  Point counting in time polynomial
in \(\log p\) distinguishes this curve from its quadratic twist: we select the
one with \((p+1)^2\) rational points.  Standard arithmetic in finite fields
then recovers \(\nu\), with randomized polynomial time when roots are found by
random sampling.

Conjugation transfers an answer on \(E_\star\) to the original curve: if
\(\alpha_\star\in\End(E_\star)\), then
\(\nu^{-1}\circ\alpha_\star\circ\nu\in\End(E_{\mathrm{in}})\).  This
conjugate descends to \(\F_{p^2}\), since the Galois cocycle of the quadratic
twist takes values in the central subgroup \(\{\pm1\}\) and fixes every
conjugated endomorphism.  Conjugation also preserves non-scalarity, and the
bounded degrees of \(\nu\) and \(\nu^{-1}\) preserve efficient
representability, so the search may start at \(E_\star\).

\subsection{Sampling distribution}

The supersingular $2$-isogeny graph, with its natural automorphism weights, has
stationary distribution
\begin{equation}\label{eq:stationary}
 \pi(E)=\frac{24}{(p-1)\#\operatorname{Aut}(E)}.
\end{equation}
A generic vertex therefore has mass $12/(p-1)$, while the Ramanujan bound gives
nontrivial normalized eigenvalues of absolute value at most $2\sqrt2/3$.
After making the walk lazy, the corresponding bound is smaller than $35/36$;
see~\cite{PW24}.

Let $\mu_t$ be the distribution after $t$ lazy steps from \(E_\star\), where
at each step the walk stays put with probability \(1/2\) and otherwise
chooses one of the three outgoing \(2\)-isogenies uniformly.  Since
$\pi_{\min}>1/p$,
\begin{equation}\label{eq:mixing}
 \|\mu_t-\pi\|_{\rm TV}
 \leq\frac{\sqrt p}{2}\left(\frac{35}{36}\right)^t.
\end{equation}
The choice
\begin{equation}\label{eq:walk-length}
 T=\ceil{\frac{(n/2+81)\log2+4\log n}{\log(36/35)}}.
\end{equation}
makes the bound in~\eqref{eq:mixing} at most \(2^{-82}n^{-4}\).
Theorem~\ref{thm:density} and \eqref{eq:stationary} therefore imply
\begin{equation}\label{eq:good-probability}
 \Pr[\exists d\in\DD_p:\ r_d(E)>0]
 \geq2^{-80}n^{-4}.
\end{equation}
The stationary mass of the set counted by Theorem~\ref{thm:density} is at
least \(12p/((p-1)2^{80}n^4)\), whereas the mixing error is smaller than
\(2^{-82}n^{-4}\).  Consequently, the expected number of trials is
\(O(n^4)\), and computing and storing each walk costs a polynomial in \(n\).

\subsection{Factor splitting}

Once the walk reaches a curve counted by Theorem~\ref{thm:density}, an
isogeny from that curve to its conjugate has degree in \(\DD_p\).  The search
splits its prime factors into two products of similar size, with parameters
\[
 Q=2^{s+4},\qquad B=\max_{d\in\DD_p}d,\qquad
 L=\ceil{\sqrt{QB}}.
\]

\begin{lemma}[Factor splitting]\label{lem:balanced}
Every $d\in\DD_p$ has a factorization $d=uv$ such that
\[
 \gcd(u,v)=1,\qquad u,v\leq L,
\]
and both $u$ and $v$ are squarefree products of primes from the interval
\eqref{eq:prime-interval}.
\end{lemma}

\begin{proof}
Process the prime factors of $d$ in any order, multiplying at each step the
smaller of two products by the next factor.  The ratio of the larger product
to the smaller one is then at most the largest factor processed so far and
hence at most $Q$, which at the end gives
$\max(u,v)^2\leq Qd\leq QB$.  Assigning each prime to only one product also
gives $\gcd(u,v)=1$.
\end{proof}

The lemma reduces the search to two lists, obtained by enumerating from
\(E\) and \(E^{(p)}\) all cyclic isogenies of degree at most \(L\) whose
degrees are squarefree products of primes from~\eqref{eq:prime-interval}.
Because the Frobenius of degree \(p^2\) is the scalar \(-p\), every subgroup
of order prime to \(p\) is stable under Galois; the quotient maps are therefore
defined over \(\F_{p^2}\), and their targets are maximal.

For every target \(C\), store a key \(\kappa(C)\) that identifies its
isomorphism class over \(\F_{p^2}\), consisting of its \(j\)-invariant and
its quadratic, quartic, or sextic twist class, as appropriate.  Each record
also stores the set of prime factors and the chain of quotient maps, and the
number of records in either list is at most
\begin{equation}\label{eq:record-count}
 \sum_{a\leq L}\Psi(a)=O(L^2).
\end{equation}
The bound follows from
$\Psi(a)=a\sum_{r\mid a}\mu^2(r)/r$ by summing first over $r$; restricting the
allowed factors can only reduce the number of records.

Standard routines based on V\'elu formulas or modular polynomials enumerate
each extension of prime degree in time polynomial in \(Q\) and \(n\), while
standard isomorphism tests over finite fields recover an isomorphism between
curves with equal keys in probabilistic polynomial time.  In particular, the
key \(\kappa\) distinguishes the nonisomorphic twists over \(\F_{p^2}\) that
share the same \(j\)-invariant.

\begin{algorithm}[H]
\caption{Isogenies from one curve:
  \(\operatorname{ListIsogenies}(E,L,\mathcal Q)\)}
\label{alg:half-isogenies}
\begin{algorithmic}[1]
\Require A supersingular curve \(E\), a bound \(L\), and the ordered prime set
         \(\mathcal Q=\{q_1,\ldots,q_m\}\)
\Ensure One record for every cyclic squarefree isogeny supported on
        \(\mathcal Q\) and of degree at most \(L\)
\State \(\mathcal L\gets\varnothing\)
\ForAll{squarefree \(a\leq L\) with \(\Supp(a)\subseteq\mathcal Q\)}
  \ForAll{cyclic subgroups \(G\subset E[a]\) of order \(a\)}
    \State construct the quotient chain \(\lambda:E\to E/G\), processing
           \(\Supp(a)\) in increasing order
    \State append
      \((\kappa(E/G),E/G,|\Supp(a)|,\Supp(a),\lambda)\) to \(\mathcal L\)
  \EndFor
\EndFor
\State \Return \(\mathcal L\)
\end{algorithmic}
\end{algorithm}

Processing the primes in a fixed order prevents several permutations of one
chain from representing the same cyclic kernel.  The procedure is an
abstract description of the usual enumeration by a tree of isogenies.  The
complexity analysis counts every record and allows a polynomial amount of
work in \(Q\) and \(n\) for every edge of prime degree.

A meeting of records
\[
 \lambda:E\longrightarrow C,\qquad
 \rho:E^{(p)}\longrightarrow C
\]
gives
\begin{equation}\label{eq:meet-composition}
 \varphi=\widehat\rho\circ\lambda:E\longrightarrow E^{(p)}.
\end{equation}
If the two supports are disjoint and have total cardinality $k$, then
$\deg\varphi\in\DD_p$ and its kernel is cyclic.  An isomorphism is inserted
when the records use different models of the same endpoint.

\subsection{List matching}

A target curve can occur many times in both lists, and testing every pair of
records with that target would exceed the claimed complexity.  Color
coding~\cite{AYZ95} selects compatible records without forming this Cartesian
product.

Independently color each prime in~\eqref{eq:prime-interval} left or right.
Retain on the $E$ side only records all of whose factors are colored left, and
on the $E^{(p)}$ side only records all of whose factors are colored right.
Group the retained records by target and by the number of prime factors.
Every factor on the first side now has the left color, while every factor on
the second side has the right color, so their sets of factors are disjoint and
only one record is needed for each target and each possible number of factors.

Fix a $(d,+1)$-structure with $d\in\DD_p$ and a factorization from
Lemma~\ref{lem:balanced}.  The probability that all $k$ factors receive the
prescribed colors is $2^{-k}$.  With
\begin{equation}\label{eq:color-repetitions}
 R=\ceil{2^k\log4}
\end{equation}
independent colorings, the two chosen factors are placed on the required sides
with probability at least $3/4$.  Because the lists of isogenies are reused,
each new coloring requires only a scan of their records and new hash tables.

Before accepting a proposed match, the algorithm checks each quotient map in
the two chains, verifies the isomorphism between the targets on their curve
equations, and checks both sets of factors.  Opposite colors make these sets
disjoint, and their union must contain \(k\) primes; the composite is therefore
a separable isogeny of squarefree degree in \(\DD_p\) with cyclic kernel.

These checks operate directly on the stored chains, with a rational map of
degree \(p^{1/3+o(1)}\) represented compositionally.  Since $4B<p$,
Proposition~\ref{prop:automatic} shows that the composite has sign \(+1\).
A failed matching attempt triggers a restart, while every returned map has
passed all checks.

\begin{algorithm}[H]
\caption{A separable isogeny to the conjugate}
\label{alg:conjugate-isogeny}
\begin{algorithmic}[1]
\Require A normalized supersingular model \(E/\F_{p^2}\)
\Ensure A separable cyclic isogeny \(\varphi:E\to E^{(p)}\) defining a
        \((d,+1)\)-structure for some \(d\in\DD_p\), or \(\bot\)
\State construct \(s,k,\mathcal Q,Q,B,L\) from
       \eqref{eq:parameters}, \eqref{eq:prime-interval}, and
       \eqref{eq:degree-family}
\State \(\mathcal L\gets\Call{ListIsogenies}{E,L,\mathcal Q}\)
\State \(\mathcal R\gets\Call{ListIsogenies}{E^{(p)},L,\mathcal Q}\)
\For{\(r=1,\ldots,\ceil{2^k\log4}\)}
  \State independently color every \(q\in\mathcal Q\) left or right
  \State make a hash table containing one record
         \((\kappa(C),C,a,S,\lambda)\in\mathcal L\) for every key
         \((\kappa(C),a)\) with all primes in \(S\) colored left
  \ForAll{\((\kappa(C'),C',b,T,\rho)\in\mathcal R\) with all primes in \(T\)
          colored right}
    \If{the table contains a record with key \((\kappa(C'),k-b)\)}
      \State retrieve \((\kappa(C),C,k-b,S,\lambda)\) and an isomorphism
             \(\iota_C:C\to C'\)
      \State \(\varphi\gets\widehat\rho\circ\iota_C\circ\lambda\)
      \If{both chains and \(\iota_C\) verify, and
          \(|S|+|T|=k\)}
        \State \Return \(\varphi\)
      \EndIf
    \EndIf
  \EndFor
\EndFor
\State \Return \(\bot\)
\end{algorithmic}
\end{algorithm}

\subsection{Complexity}

Equations~\eqref{eq:sk-size},~\eqref{eq:B-size}, and the definition of $Q$
give
\begin{align*}
 L^2&\leq4QB
 =p^{1/3}\exp\bigl(O(\sqrt{\log p\,\log\log p})\bigr),\\
 2^k&=\exp\bigl(O(\sqrt{\log p/\log\log p})\bigr),\\
 \operatorname{poly}(Q,n)
 &=\exp\bigl(O(\sqrt{\log p\,\log\log p})\bigr).
\end{align*}
Consequently one trial at a sampled curve, including all colorings, costs
\begin{equation}\label{eq:trial-cost}
p^{1/3}\exp\bigl(O(\sqrt{\log p\,\log\log p})\bigr)
\end{equation}
time and at most the same amount of memory.  The \(O(n^4)\) expected trials
from~\eqref{eq:good-probability} are absorbed by the displayed bound.

\subsection{Endomorphism construction}

For a recovered conjugate isogeny define
\begin{equation}\label{eq:one-end}
 \alpha_E=\pi_{E^{(p)}}\circ\varphi\in\End(E)\setminus\Z.
\end{equation}
Its inseparable degree is $p$, whereas the valuation at \(p\) of the
inseparable degree of a scalar multiplication is even; hence $\alpha_E$ is
non-scalar.

Let $\omega:E_\star\to E$ be the separable part of the sampled walk, including
the final isomorphism to the fixed model.  Then
\begin{equation}\label{eq:transport}
 \alpha_\star
 =\widehat\omega\circ\alpha_E\circ\omega
 \in\End(E_\star)\setminus\Z
\end{equation}
is non-scalar: after division by the central scalar $\deg\omega$, it is the
conjugate of $\alpha_E$ in $\End^0(E_\star)$.  If the walk contains \(r\leq T\)
nontrivial steps, then
\begin{equation}\label{eq:output-degree}
 \deg\alpha_\star=2^{2r}p\deg\varphi\leq2^{2T}pB,
 \qquad \log\deg\alpha_\star=O(n).
\end{equation}

\subsection{Output representation}
\label{sec:output-representation}

\begin{lemma}[Representation of the recovered map]\label{lem:compression}
The chain representing \(\varphi\) can be converted, within the time bound
\eqref{eq:trial-cost}, to an efficient representation, and consequently
\(\nu^{-1}\circ\alpha_\star\circ\nu\) has an efficient representation of
polynomial length in \(n\).
\end{lemma}

\begin{proof}
Let \(d=\deg\varphi\).  IsogenyInterpolation requires an integer \(N>d\)
that is coprime to \(pd\), bases for the components of \(E[N]\) at each prime,
and their images under \(\varphi\) \cite[Proposition~2]{MW25}.  We can choose
such a squarefree integer \(N\) with
\begin{equation}\label{eq:interpolation-modulus}
 \log N=O(n),\qquad P^+(N)=O(n),
\end{equation}
where \(P^+\) denotes the largest prime factor.  To construct \(N\), take the
product of the primes up to \(Cn\), for a sufficiently large absolute
constant \(C\), and remove the primes that divide \(pd\).  Chebyshev's
estimates~\cite{RS62} show that the logarithm of the original product is
linear in \(Cn\), whereas the removed primes have total logarithm at most
\(\log(pd)=O(n)\); increasing \(C\) therefore leaves a product larger than
\(d\).

On a maximal model, the Frobenius of degree \(p^2\) is the scalar \(-p\).
For each prime \(\ell\mid N\), a basis of \(E[\ell]\) is therefore defined
over an extension of degree at most \(\ell-1=O(n)\) and can be computed in
time polynomial in \(n\) and \(\ell\); compare~\cite[Lemma~6]{MW25}.

Evaluate the stored chain for \(\varphi\) on these bases, and then invoke
IsogenyInterpolation.  There are only polynomially many basis points, so the
evaluations multiply the search cost by a polynomial in \(n\).  The
interpolation is polynomial in its input length and in \(P^+(N)\), and both
costs are absorbed by~\eqref{eq:trial-cost}.

The map \(\pi_{E^{(p)}}\), the \(O(n)\) factors of degree two in \(\omega\)
and \(\widehat\omega\), and the bounded degree isomorphisms
\(\nu,\nu^{-1}\) all have efficient evaluators.  Composing these evaluators
with the representation of \(\varphi\) gives an efficient representation of
\(\nu^{-1}\circ\alpha_\star\circ\nu\).  All operations use the input curve,
the stored chains, finite-field arithmetic, and IsogenyInterpolation.
\end{proof}

\begin{algorithm}[H]
\caption{A non-scalar endomorphism}
\label{alg:one-end}
\begin{algorithmic}[1]
\Require A supersingular elliptic curve \(E_{\mathrm{in}}/\F_{p^2}\)
\Ensure A non-scalar
        \(\alpha_{\mathrm{in}}\in\End(E_{\mathrm{in}})\setminus\Z\) in
        efficient representation
\State \(n\gets\floor{\log_2p}+1\) and compute \(T\) from
       \eqref{eq:walk-length}
\If{\(j(E_{\mathrm{in}})=0\) or \(1728\)}
  \State \Return an explicit automorphism of order \(3\) or \(4\)
\EndIf
\State compute the maximal model \(E_\star\) and the isomorphism
       \(\nu:E_{\mathrm{in}}\to E_\star\) from
       \eqref{eq:input-isomorphism}
\Loop
  \State take a lazy \(2\)-isogeny walk
         \(\omega:E_\star\to E\) of length \(T\), retaining its chain
  \If{\(j(E)=0\) or \(1728\)}
    \State \textbf{continue}
  \EndIf
  \State replace \(E\) by the fixed normalized model of its invariant and
         append the corresponding isomorphism to \(\omega\)
  \State \(\varphi\gets\Call{IsogenyToConjugate}{E}\)
  \If{\(\varphi\neq\bot\)}
    \State convert \(\varphi\) to an efficient representation using
           Lemma~\ref{lem:compression}
    \State \(\alpha_E\gets\pi_{E^{(p)}}\circ\varphi\)
    \State \(\alpha_\star\gets
           \widehat\omega\circ\alpha_E\circ\omega\)
    \State \(\alpha_{\mathrm{in}}\gets
           \nu^{-1}\circ\alpha_\star\circ\nu\)
    \State \Return \(\alpha_{\mathrm{in}}\)
  \EndIf
\EndLoop
\end{algorithmic}
\end{algorithm}

\begin{theorem}[Main algorithm]
\label{thm:algorithm}
There is a classical Las Vegas algorithm which, given a supersingular elliptic
curve $E_{\mathrm{in}}/\F_{p^2}$, returns a non-scalar endomorphism
\(\alpha_{\mathrm{in}}\in\End(E_{\mathrm{in}})\setminus\Z\) in efficient
representation in expected time and memory
\begin{equation}\label{eq:main-complexity}
 \boxed{
p^{1/3}\exp\bigl(O(\sqrt{\log p\,\log\log p})\bigr)
 =p^{1/3+o(1)}.}
\end{equation}
The output satisfies \(\log\deg\alpha_{\mathrm{in}}=O(\log p)\), and the
analysis is unconditional and independent of smoothness or statistical
independence heuristics.
\end{theorem}

\begin{proof}
The two special invariants terminate in polynomial time, and
\cref{sec:input-normalization} reduces every other input to a generic maximal
model without changing the asserted exponent.  Theorem~\ref{thm:density} and
the mixing estimate give the success probability
\eqref{eq:good-probability}.  Conditional on reaching a curve in
the set from Theorem~\ref{thm:density}, choose one supported
\((d,+1)\)-structure and split its degree as in Lemma~\ref{lem:balanced}.
Algorithm~\ref{alg:half-isogenies} contains the two isogenies determined by
this split, and the analysis of the coloring step shows that
Algorithm~\ref{alg:conjugate-isogeny} recovers a valid match with probability
at least \(3/4\).

Every returned map passes the deterministic checks described above, while an
unsuccessful coloring or sampled curve causes repetition, so correctness is
deterministic.  Equation~\eqref{eq:trial-cost} and the polynomial expected
number of trials give~\eqref{eq:main-complexity}.  Equations
\eqref{eq:one-end}--\eqref{eq:output-degree} prove that the output of
Algorithm~\ref{alg:one-end} is a non-scalar endomorphism of the original curve,
and Lemma~\ref{lem:compression} gives the required efficient representation.

The density theorem assumes \(n\geq2^{16}\), leaving only finitely many
smaller input sizes.  An exhaustive search in the supersingular isogeny graph
handles them, and enlarging the absolute constant in the
\(O(\,\cdot\,)\) notation absorbs their cost, proving the theorem for every
prime \(p>3\).
\end{proof}

\begin{corollary}\label{cor:equivalent-problems}
The supersingular endomorphism ring problem and the supersingular isogeny
problem admit classical probabilistic algorithms with the same
$p^{1/3+o(1)}$ time and memory bound.
\end{corollary}

\begin{proof}
Equation~\eqref{eq:output-degree} solves the bounded problem
\(\OneEnd_\lambda\) with \(\lambda(n)=O(n)\).  The reduction from
\EndRing\ to bounded \OneEnd\ of Page--Wesolowski and the unconditional
equivalences of Herl\'edan Le Merdy--Wesolowski
\cite[Theorem~7.1]{PW24}\cite[Theorem~1.1]{MW25} have overhead polynomial in
\(n\) and \(\lambda(n)\), which is absorbed by
\eqref{eq:main-complexity}.  This corollary concerns \Isogeny\ with
unrestricted output degree; a path whose prime degree is prescribed in
advance lies outside its statement.
\end{proof}

\section{Conclusion}

We have proved an unconditional \(p^{1/3+o(1)}\) bound for the time and memory
of \OneEnd; the known reductions give the same bound for \EndRing\ and
\Isogeny.  This result attains the exponent \(1/3\) without the smoothness
assumption used in the previous algorithm.

The main step is the collision bound of Theorem~\ref{thm:collision}.  Two
isogenies from one curve to its conjugate determine oriented embeddings of
quadratic orders, and squarefree degrees prevent a path from descending and
then ascending at the same conductor prime, so the embeddings can be compared
inside a common order.  The possible coincidences lie in a fiber of the square
map on its class group, controlled by the \(2\)-torsion, while primes shared
by the two degrees contribute the factor
\[
 \prod_{\ell\mid\gcd(d,e)}\left(1+\frac3\ell\right).
\]

This collision bound turns the Chenu--Smith count into a statement about many
distinct curves.  Because the degree family is chosen before sampling and
every degree in it is a product of small primes, the algorithm can enumerate
the required isogenies from their known factorizations.

Potential refinements concern the constants and memory use: a more detailed
analysis of ideals above \(2\) and of ambiguous forms would reduce the
constant in the collision bound, while a different method for matching the
two lists might reduce the \(p^{1/3+o(1)}\) memory requirement, with both
changes preserving the unconditional exponent established here.

\bibliography{references}

\end{document}